\documentclass[a4paper,USEnglish,cleveref, autoref]{lipics-v2021}
\pdfoutput=1 %
\hideLIPIcs  %
\nolinenumbers

\usepackage{float}
\usepackage{listings}
\usepackage{xcolor}
\usepackage[utf8]{inputenc} %
\usepackage[T1]{fontenc}    %
\usepackage{booktabs}       %
\usepackage{amsfonts}       %
\usepackage{nicefrac}       %
\usepackage{microtype}      %
\usepackage{mathtools}
\usepackage{amsmath}
\usepackage{amssymb}
\usepackage{amsthm}

\usepackage{dsfont}

\newcommand{\Ab}{\ensuremath{\mathbf{A}}}

\newcommand{\Bb}{\ensuremath{\mathbf{B}}}

\newcommand{\ot}{\ensuremath{\otimes}}

\newcommand{\bp}{\mathsf{bp}}
\newcommand{\bc}{\mathsf{bc}}
\newcommand{\bpf}{\mathsf{bp}_f}
\newcommand{\bcf}{\mathsf{bc}_f}
\usepackage{pgf,tikz,pgfplots}
\pgfplotsset{compat=1.15}
\usepackage{mathrsfs}

\usetikzlibrary{fit,shapes.callouts,shapes.arrows,shapes.geometric,decorations.pathmorphing,positioning}
\usetikzlibrary{decorations,decorations.pathmorphing,decorations.pathreplacing}
\usetikzlibrary{plotmarks,calc}
\usetikzlibrary{positioning,automata,arrows}

\tikzset{
  every overlay node/.style={
    anchor=north west,
  },
}

\title{Engineering Insights into Biclique Partitions and Fractional Binary Ranks of Matrices}
\author{Angikar Ghosal}{%
  Graduate School of Business, 
  Stanford University, 
  655 Knight Way, 
  Stanford, CA 94305}{angikar@stanford.edu}{https://orcid.org/0009-0002-9336-5342}{}
  
\author{Andreas Karrenbauer}{%
    Max Planck Institut for Informatics, Saarland Informatics Campus,
    66123 Saarbrücken, Germany \and \url{https://people.mpi-inf.mpg.de/~karrenba}}
    {andreas.karrenbauer@mpi-inf.mpg.de}
    {https://orcid.org/0000-0001-6129-3220}{}

\authorrunning{A.\ Ghosal and A.\ Karrenbauer} %

\Copyright{Angikar Ghosal and Andreas Karrenbauer} %


\makeatletter%
\newcommand{\@noticestring}{}
\makeatother
\begin{document}
\maketitle

\begin{abstract}
We investigate structural properties of the binary rank of Kronecker powers of binary matrices, equivalently, the biclique partition numbers of the corresponding bipartite graphs. To this end, we engineer a Column Generation approach to solve linear optimization problems for the fractional biclique partition number of bipartite graphs, specifically examining the Domino graph and its Kronecker powers.
We address the challenges posed by the double exponential growth of the number of bicliques in increasing Kronecker powers.
We discuss various strategies to generate suitable initial sets of bicliques, including an inductive method for increasing Kronecker powers. We show how to manage the number of active bicliques to improve running time and to stay within memory limits.
Our computational results reveal that the fractional binary rank is not multiplicative with respect to the Kronecker product. Hence, there are binary matrices, and bipartite graphs, respectively, such as the Domino, where the asymptotic fractional binary rank is strictly smaller than the fractional binary rank. While we used our algorithm to reduce the upper bound, we formally prove that the fractional biclique cover number is a lower bound, which is at least as good as the widely used isolating (or fooling set) bound. For the Domino, we obtain that the asymptotic fractional binary rank lies in the interval $[2,2.373]$. Since our computational resources are not sufficient to further reduce the upper bound, we encourage further exploration using more substantial computing resources or further mathematical engineering techniques to narrow the gap and advance our understanding of biclique partitions, particularly, to settle the open question whether binary rank and biclique partition number are multiplicative with respect to the Kronecker product.

\end{abstract}

\section{Introduction}

Algorithm engineering is often used in an applied context or to engineer algorithms that perform well on practical instances. In this paper, we use algorithm engineering to make progress on a structural mathematical question regarding the \emph{asymptotic (fractional) binary rank} of binary matrices, which is related to the (fractional) biclique partition number of bipartite graphs. This problem has applications to communication complexity and information theory~\cite{Chee2024}.

Whenever we have a rank function, say $r: \mathbb{R}^{n\times m} \to \mathbb{R}$, that is sub-multiplicative with respect to the Kronecker product, i.e., $r(A \otimes A') \le r(A) \cdot r(A')$ for matrices $A$ and $A'$, the corresponding asymptotic rank
\[
r^\infty(A) := \lim_{k\rightarrow\infty} \sqrt[k]{r(A^{\otimes k})}
\]
is well-defined due to Fekete's Lemma~\cite{fekete1923verteilung}.

The question whether the asymptotic rank is equal to the rank, which holds if the rank function is multiplicative with respect to the Kronecker product, has received considerable attention in the past for various rank functions. It is well-known that the ordinary rank over the reals is multiplicative, thus, the asymptotic real rank is equal to the real rank. On the other hand, the asymptotic nonnegative rank can be strictly less than the nonnegative rank of a nonnegative real matrix~\cite{Beasley2013,Chee2024}. This is also true for the so-called boolean rank~\cite{watts2001boolean,Haviv2021}, whereas its fractional relaxation is multiplicative~\cite{watts2006fractional}, i.e., the asymptotic fractional boolean rank always coincides with the fractional boolean rank.

A big and longstanding open question is whether the binary rank is multiplicative~\cite{Watson2016}. Although we do not answer this question definitively, in this paper, we shed some light on it by investigating its fractional relaxation. Since the fractional boolean rank is multiplicative it is tempting to conjecture the same for the fractional binary rank. We disprove this conjecture with the explicit example of a $3 \times 3$ binary matrix that is connected to the so-called \emph{Domino} graph. Moreover, we provide upper and lower bounds for its asymptotic rank. While the lower comes from its fractional boolean rank and is proven mathematically, the upper bound is derived computationally by solving linear optimization problems with exponentially many variables. To this end, we engineered a column generation approach that allows us to solve the linear programs for higher and higher Kronecker powers bringing down the upper bound from $2.5$ to about $2.373$ while our lower bound is $2$. We admit that the gap is still large, but we hope that our approach inspires other with larger computational resources (particularly more than 256 GB RAM) or by further engineering (e.g., size reduction by exploiting symmetries) to pick up the challenge and close the gap further.

\section{Preliminaries and Notation}

Consider an $m\times n$ matrix $A$, where entries are $0$ or $1$, i.e., $A_{ij}\in\{0,1\}$ and the following alternative characterization of its \emph{real} rank, i.e., its rank over the reals:
\[
rank(A) = \min \{ r : \exists U \in \mathbb{R}^{n \times r}, V \in \mathbb{R}^{r \times m} \text{ s.t.\ } U \cdot V = A \}
\]
Similarly, we have the following quantities defined \cite{watts2001boolean,Watson2016,haviv2023binarybooleanrankregular}:
\begin{enumerate}
\item The non-negative rank $rank_{\ge 0}(A)$, where the factors are required to be non-negative matrices, i.e., 
\[
rank_{\ge 0}(A) = \min \{ r : \exists U \in \mathbb{R}_{\ge 0}^{n \times r}, V \in \mathbb{R}_{\ge 0}^{r \times m} \text{ s.t.\ } U \cdot V = A \}.
\]
\item The binary rank $rank_{01}(A)$, where the factors are required to be binary matrices, i.e.,
\[
rank_{01}(A) = \min \{ r : \exists U \in \{0,1\}^{n \times r}, V \in \{0,1\}^{r \times m} \text{ s.t.\ } U \cdot V = A \}.
\]
\item The Boolean rank $rank_{01\updownarrow}(A)$, where Boolean arithmetic or some spread is used, i.e.,
\[
rank_{01\updownarrow}(A) = \min \{ r : \exists U \in \{0,1\}^{n \times r}, V \in \{0,1\}^{r \times m} \text{ s.t.\ } A \le U \cdot V \le r \cdot A \},
\]
where $A \le U\cdot V \le r \cdot A$ is considered entry-wise.
\end{enumerate}
There is a combinatorial meaning for the latter two quantities. Consider the bipartite graph $G_A = (V, E)$  with $V :=  [n+m]$ and $E := \left\{\{i,n+j\}: A_{ij}=1\right\}$. Let $u \in \{0,1\}^n$ and $v \in \{0,1\}^m$. If the binary rank-one matrix $u v^T$ is upper bounded by $A$ entry-wise, i.e., $u v^T \le A$, then it induces a complete bipartite subgraph, also called a \emph{biclique}, of $G_A$. Moreover, a binary $r$-decomposition of $A$, i.e., $U \in \{0,1\}^{n \times r}$ and $V \in \{0,1\}^{r \times m}$  with $U\cdot V = A$ represents a partition of the set of edges of $G_A$ into $r$ bicliques. If $r = rank_{01}(A)$, then it is the biclique partition with the smallest number of bicliques and $r$ is called the \emph{biclique partition number} $\bp(G_A)$. Hence, $rank_{01}(A) = \bp(G_A) = \bp(A)$ for short. Similarly, $rank_{01\updownarrow}(A) = \bc(G_A) = \bc(A)$, the \emph{biclique cover number} of $G_A$, where an edge of $G_A$ may be contained in more than one biclique. Hence, both quantitites are the optimum objective values of combinatorial optimization problems that can be modeled using Integer Linear Optimization. To this end, consider the binary edge-biclique-incidence matrix $M$, which has a row for each edge and a column for each of the $N$ bicliques of $G_A$. This yields
\begin{equation}
\bp(A) = \min \left\{ \mathds{1}^T x : \> M x = \mathds{1}, \> x \in \{0,1\}^N \right\}\label{eqn:bp}
\end{equation}
and 
\begin{equation}
    \bc(A) = \min \left\{ \mathds{1}^T x : \> M x \ge \mathds{1}, \> x \in \{0,1\}^N \right\},\label{eqn:bc}
\end{equation}
respectively, where $\mathds{1}$ denotes an all-ones vector of appropriate dimension.
As customary in integer optimization, we also consider the corresponding linear relaxations, i.e.,
\begin{equation}
\bpf(A) = \min \left\{ \mathds{1}^T x : \> M x = \mathds{1}, \> x \ge 0 \right\}\label{eqn:bpf}
\end{equation}
and 
\begin{equation}
\bcf(A) = \min \left\{ \mathds{1}^T x : \> M x \ge \mathds{1}, \> x \ge 0 \right\},\label{eqn:bcf}
\end{equation}
respectively, where $\bpf$ and $\bcf$ are called the \emph{fractional} biclique partition and cover numbers. Correspondingly, we can define the fractional binary and Boolean ranks. That is,
\begin{align*}
   \bpf(A) &= rank_{f01}(A) \\
    &= \min \{ tr(\Lambda): \exists U \in \{0,1\}^{n \times r}, \Lambda \in diag([0,1]^r), V \in \{0,1\}^{r \times m} \text{ s.t.\ } U \Lambda V = A \},
\end{align*}
and 
\begin{align*}
    \bcf(A) &= rank_{f01\updownarrow}(A) \\
    &= \min \{ tr(\Lambda): \exists U \in \{0,1\}^{n \times r}, \Lambda \in diag([0,1]^r), V \in \{0,1\}^{r \times m} \text{ s.t.\ } A \le U \Lambda V \le r\cdot A \},
\end{align*}
respectively.

Suppose, for a bipartite graph that has $m$ and $n$ vertices in its two parts, we have a biclique that contains $m^*$ and $n^*$ vertices in its two parts, $m^*\leq m, n^*\leq n$. As all pairs of edges exist between these $m^*$ and $n^*$ vertices, one could obtain $(2^{m^*}-1)\cdot (2^{n^*}-1)$ different subbicliques from this biclique. Each of these subbicliques could belong to a bigger biclique.

Note that the biclique cover problem can be solved by using only inclusion-wise maximal bicliques, i.e., we may only consider maximal bicliques. 
Likewise, the fractional biclique cover problem can be solved accurately even if we only consider the set of maximal bicliques.

However, for the biclique partition problem (and the fractional biclique partition problem), we cannot just work with the maximal bicliques. The problem might even become infeasible. To see this, consider the following example of the so-called Domino graph with its maximal bicliques:
\[
D=\begin{bmatrix}
    1&1&0\\1&1&1\\0&1&1
\end{bmatrix};
\qquad
\mathcal{B} =
\left\{
\begin{bmatrix}
    1&1&0\\1&1&0\\0&0&0
\end{bmatrix},
\>
\begin{bmatrix}
    0&0&0\\0&1&1\\0&1&1
\end{bmatrix},
\>
\begin{bmatrix}
    0&1&0\\0&1&0\\0&1&0
\end{bmatrix},
\>
\begin{bmatrix}
    0&0&0\\1&1&1\\0&0&0
\end{bmatrix}
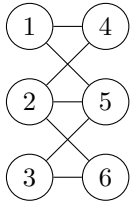
\right\}
\]
\begin{center}
\hspace{4mm}
\begin{tikzpicture}
  [
    every node/.style={draw, circle},
    node distance=1cm
  ]

  \node (0) {1};
  \node[right of=0] (3) {4};
  \node[below of=0] (1) {2};
  \node[right of=1] (4) {5};
  \node[below of=1] (2) {3};
  \node[right of=2] (5) {6};

  \draw (0) -- (3);
  \draw (0) -- (4);
  \draw (1) -- (3);
  \draw (1) -- (4);
  \draw (1) -- (5);
  \draw (2) -- (4);
  \draw (2) -- (5);
\end{tikzpicture}
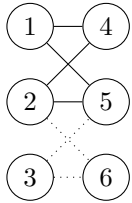
\hspace{17mm}
\begin{tikzpicture}
  [
    every node/.style={draw, circle},
    node distance=1cm
  ]

  \node (0) {1};
  \node[right of=0] (3) {4};
  \node[below of=0] (1) {2};
  \node[right of=1] (4) {5};
  \node[below of=1] (2) {3};
  \node[right of=2] (5) {6};

  \draw (0) -- (3);
  \draw (0) -- (4);
  \draw (1) -- (3);
  \draw (1) -- (4);
  \draw[dotted] (1) -- (5);
  \draw[dotted] (2) -- (4);
  \draw[dotted] (2) -- (5);
\end{tikzpicture}
\phantom{,}\>
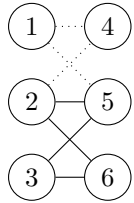
\begin{tikzpicture}
  [
    every node/.style={draw, circle},
    node distance=1cm
  ]

  \node (0) {1};
  \node[right of=0] (3) {4};
  \node[below of=0] (1) {2};
  \node[right of=1] (4) {5};
  \node[below of=1] (2) {3};
  \node[right of=2] (5) {6};

  \draw[dotted] (0) -- (3);
  \draw[dotted] (0) -- (4);
  \draw[dotted] (1) -- (3);
  \draw (1) -- (4);
  \draw (1) -- (5);
  \draw (2) -- (4);
  \draw (2) -- (5);
\end{tikzpicture}
\phantom{,}\>
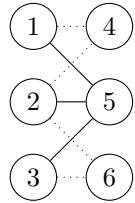
\begin{tikzpicture}
  [
    every node/.style={draw, circle},
    node distance=1cm
  ]

  \node (0) {1};
  \node[right of=0] (3) {4};
  \node[below of=0] (1) {2};
  \node[right of=1] (4) {5};
  \node[below of=1] (2) {3};
  \node[right of=2] (5) {6};

  \draw[dotted] (0) -- (3);
  \draw (0) -- (4);
  \draw[dotted] (1) -- (3);
  \draw (1) -- (4);
  \draw[dotted] (1) -- (5);
  \draw (2) -- (4);
  \draw[dotted] (2) -- (5);
\end{tikzpicture}
\phantom{,}\>
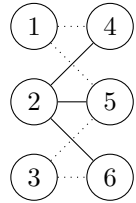
\begin{tikzpicture}
  [
    every node/.style={draw, circle},
    node distance=1cm
  ]

  \node (0) {1};
  \node[right of=0] (3) {4};
  \node[below of=0] (1) {2};
  \node[right of=1] (4) {5};
  \node[below of=1] (2) {3};
  \node[right of=2] (5) {6};

  \draw[dotted] (0) -- (3);
  \draw[dotted] (0) -- (4);
  \draw (1) -- (3);
  \draw (1) -- (4);
  \draw (1) -- (5);
  \draw[dotted] (2) -- (4);
  \draw[dotted] (2) -- (5);
\end{tikzpicture}
\end{center}

Note that every biclique partition is also a biclique cover, so the biclique cover number is always at most the biclique partition number. Likewise, every integer biclique partition (or integer biclique cover) also satisfies the constraints of a fractional biclique partition (or fractional biclique cover), so the fractional biclique partition number is less than or equal to the biclique partition number, and the fractional biclique cover number is less than or equal to the biclique cover number. A lower bound for the fractional biclique cover number is the maximum size of an induced matching~\cite{watts2006fractional}, which is equal to the maximum number of ones $i(A)$ in an \emph{isolated} submatrix of $A$, i.e., where no pair of ones is contained in the same all-ones submatrix of $A$. Such a set of ones is also called fooling set.

For the Domino graph, which is the smallest graph where biclique partition and cover number are different, the fractional biclique partition number is $2.5$, which is given by the following decomposition:
\[
\begin{bmatrix}
    1&1&0\\1&1&1\\0&1&1
\end{bmatrix}
=
\frac{1}{2}
\begin{bmatrix}
    1&1&0\\1&1&0\\0&0&0
\end{bmatrix}
+\frac{1}{2}
\>
\begin{bmatrix}
    0&0&0\\0&1&1\\0&1&1
\end{bmatrix}
+
\frac{1}{2}
\begin{bmatrix}
    1&1&0\\0&0&0\\0&0&0
\end{bmatrix}
+
\frac{1}{2}
\begin{bmatrix}
    0&0&0\\1&0&1\\0&0&0
\end{bmatrix}
+
\frac{1}{2}
\begin{bmatrix}
    0&0&0\\0&0&0\\0&1&1
\end{bmatrix}
\]

A witness for optimality is
\[
Y^* = 
\begin{bmatrix}
\frac{1}{2} & \frac{1}{2} & 0 \\
 \frac{1}{2} & -\frac{1}{2} & \frac{1}{2} \\
  0 & \frac{1}{2} & \frac{1}{2}
\end{bmatrix},
\]
representing the feasible solution $y^* = \left(\frac{1}{2}, \frac{1}{2}, \frac{1}{2}, -\frac{1}{2}, \frac{1}{2}, \frac{1}{2}\right)$ of 
\begin{equation}
\max \left\{ \mathds{1}^T y : \> M^T y \le \mathds{1} \right\},\label{eqn:dual}
\end{equation}
which is the dual of \eqref{eqn:bpf}. Since both objective values coincides, they must be optimal.

The Kronecker product of two matrices $A$ and $B$ can be defined as follows:
If matrix $\Ab$ is a $p\times q$ matrix with entries $a_{ij}$ and $\Bb$ is a $r\times s$ matrix with entries $b_{ij}$, then the Kronecker product $\Ab\otimes\Bb$ is defined as:
\[\Ab\otimes\Bb=\begin{bmatrix}a_{11}\Bb&\hdots&a_{1q}\Bb\\\vdots&\ddots&\vdots\\a_{p1}\Bb&\hdots&a_{pq}\Bb\end{bmatrix}\]

Since the Kronecker product of two binary matrices is a binary matrix again, there is also a corresponding bipartite graph and we can transfer the Kronecker product to bipartite graphs. Note that the term Kronecker graph product often refers also to non-bipartite graphs so that the whole adjacency matrix is used to define it and not only the reduced bipartite adjacency matrix as in our case. However, the Kronecker graph product of two bipartite graphs yields again a bipartite graph with two isomorphic connected components, which are isomorphic to the bipartite graph corresponding to the Kronecker product of the two bipartite adjacency matrices. Hence, biclique partition, biclique cover number, and their fractional relaxations are the same. Hence, we always mean the Kronecker product of the bipartite adjacency matrices in this paper. To avoid confusion, the term \emph{weak bipartite product} is used in the literature~\cite{watts2006fractional} with the following definition: Given bipartite graphs $G$ and $H$ with an ordered bi-partitions $(X_G, Y_G)$ and $(X_H, Y_H)$, the weak bipartite product of $G$ and $H$ is the bipartite graph $G\ot  H$ with ordered bi-partition $(X_G\times X_H, Y_G\times Y_H)$, where two vertices $(g_i, h_k)$ and $(g_j, h_l)$ in $G\ot  H$ are adjacent if and only if $g_i$ is adjacent to $g_j$ in $G$ and $h_k$ is adjacent to $h_l$ in $H$. If $A_G$ and $A_H$ are the two bipartite adjacency matrices of bipartite graphs $G$ and $H$, we have that the Kronecker product $A_G\otimes A_H$ is the bipartite adjacency matrix of $G\ot  H$.
\section{The Asymptotic (Fractional) Binary Rank}
It is well known for the real rank that $rank(A \otimes A') = rank(A) \cdot rank(A')$ holds for any real matrices $A$ and $A'$. However, it is a long standing open problem whether the same also holds for the binary rank:
\begin{conjecture}
Given binary matrices $A$ and $A'$, it holds that
\[
rank_{01}(A\otimes A') = rank_{01}(A) \cdot rank_{01}(A').
\]
\end{conjecture}
To see that $rank_{01}(A\otimes A') \le rank_{01}(A) \cdot rank_{01}(A')$, consider  the decompositions of $A = U\cdot V$ and $A' = U' \cdot V'$ that certify the ranks $r,r'$ of $A, A'$, respectively. Hence, $A \otimes A' = UV \otimes U'V' = (U \otimes U') \cdot (V \otimes V')$ yielding a decomposition of $A \otimes A'$ into $r\cdot r'$ bicliques because the Kronecker product of two bicliques is again a biclique. However, it is not clear whether this upper bound is always tight. The same argument also applies for the other notion of ranks. However, for the boolean rank, it was shown that equality does not always hold~\cite{watts2001boolean,Haviv2021}. In contrast, equality holds for the fractional boolean rank, i.e., fractional biclique cover number~\cite{watts2006fractional}. For the fractional binary rank, equality does not always hold as the following example with $D \otimes D$ shows, where we present a decomposition that covers each edge exactly twice (zeros are represented by dots for better readability): 
\[
\begin{array}{cccc}
\setlength\arraycolsep{3pt}
\begin{bmatrix}
1 & 1 & \cdot & 1 & \cdot & \cdot & \cdot & \cdot & \cdot\\
1 & 1 & \cdot & 1 & \cdot & \cdot & \cdot & \cdot & \cdot\\
\cdot & \cdot & \cdot & \cdot & \cdot & \cdot & \cdot & \cdot & \cdot\\
1 & 1 & \cdot & 1 & \cdot & \cdot & \cdot & \cdot & \cdot\\
1 & 1 & \cdot & 1 & \cdot & \cdot & \cdot & \cdot & \cdot\\
\cdot & \cdot & \cdot & \cdot & \cdot & \cdot & \cdot & \cdot & \cdot\\
\cdot & \cdot & \cdot & \cdot & \cdot & \cdot & \cdot & \cdot & \cdot\\
\cdot & \cdot & \cdot & \cdot & \cdot & \cdot & \cdot & \cdot & \cdot\\
\cdot & \cdot & \cdot & \cdot & \cdot & \cdot & \cdot & \cdot & \cdot\\
\end{bmatrix}
&
\setlength\arraycolsep{3pt}
\begin{bmatrix}
\cdot & \cdot & \cdot & \cdot & \cdot & \cdot & \cdot & \cdot & \cdot\\
\cdot & 1 & 1 & \cdot & 1 & 1 & \cdot & \cdot & \cdot\\
\cdot & 1 & 1 & \cdot & 1 & 1 & \cdot & \cdot & \cdot\\
\cdot & \cdot & \cdot & \cdot & \cdot & \cdot & \cdot & \cdot & \cdot\\
\cdot & 1 & 1 & \cdot & 1 & 1 & \cdot & \cdot & \cdot\\
\cdot & 1 & 1 & \cdot & 1 & 1 & \cdot & \cdot & \cdot\\
\cdot & \cdot & \cdot & \cdot & \cdot & \cdot & \cdot & \cdot & \cdot\\
\cdot & \cdot & \cdot & \cdot & \cdot & \cdot & \cdot & \cdot & \cdot\\
\cdot & \cdot & \cdot & \cdot & \cdot & \cdot & \cdot & \cdot & \cdot\\
\end{bmatrix} 
   &
\setlength\arraycolsep{3pt}
\begin{bmatrix}
1 & \cdot & \cdot & 1 & 1 & \cdot & \cdot & \cdot & \cdot\\
1 & \cdot & \cdot & 1 & 1 & \cdot & \cdot & \cdot & \cdot\\
\cdot & \cdot & \cdot & \cdot & \cdot & \cdot & \cdot & \cdot & \cdot\\
\cdot & \cdot & \cdot & \cdot & \cdot & \cdot & \cdot & \cdot & \cdot\\
\cdot & \cdot & \cdot & \cdot & \cdot & \cdot & \cdot & \cdot & \cdot\\
\cdot & \cdot & \cdot & \cdot & \cdot & \cdot & \cdot & \cdot & \cdot\\
\cdot & \cdot & \cdot & \cdot & \cdot & \cdot & \cdot & \cdot & \cdot\\
\cdot & \cdot & \cdot & \cdot & \cdot & \cdot & \cdot & \cdot & \cdot\\
\cdot & \cdot & \cdot & \cdot & \cdot & \cdot & \cdot & \cdot & \cdot\\
\end{bmatrix} 
&
\setlength\arraycolsep{3pt}
\begin{bmatrix}
\cdot & \cdot & \cdot & \cdot & \cdot & \cdot & \cdot & \cdot & \cdot\\
\cdot & \cdot & 1 & \cdot & \cdot & 1 & \cdot & \cdot & \cdot\\
\cdot & \cdot & 1 & \cdot & \cdot & 1 & \cdot & \cdot & \cdot\\
\cdot & \cdot & \cdot & \cdot & \cdot & \cdot & \cdot & \cdot & \cdot\\
\cdot & \cdot & \cdot & \cdot & \cdot & \cdot & \cdot & \cdot & \cdot\\
\cdot & \cdot & \cdot & \cdot & \cdot & \cdot & \cdot & \cdot & \cdot\\
\cdot & \cdot & \cdot & \cdot & \cdot & \cdot & \cdot & \cdot & \cdot\\
\cdot & \cdot & \cdot & \cdot & \cdot & \cdot & \cdot & \cdot & \cdot\\
\cdot & \cdot & \cdot & \cdot & \cdot & \cdot & \cdot & \cdot & \cdot\\
\end{bmatrix}
\\
\setlength\arraycolsep{3pt}
\begin{bmatrix}
\cdot & \cdot & \cdot & \cdot & \cdot & \cdot & \cdot & \cdot & \cdot\\
\cdot & \cdot & \cdot & \cdot & \cdot & \cdot & \cdot & \cdot & \cdot\\
\cdot & \cdot & \cdot & \cdot & \cdot & \cdot & \cdot & \cdot & \cdot\\
\cdot & \cdot & \cdot & \cdot & \cdot & \cdot & \cdot & \cdot & \cdot\\
\cdot & \cdot & \cdot & 1 & 1 & \cdot & 1 & \cdot & \cdot\\
\cdot & \cdot & \cdot & \cdot & \cdot & \cdot & \cdot & \cdot & \cdot\\
\cdot & \cdot & \cdot & 1 & 1 & \cdot & 1 & \cdot & \cdot\\
\cdot & \cdot & \cdot & 1 & 1 & \cdot & 1 & \cdot & \cdot\\
\cdot & \cdot & \cdot & \cdot & \cdot & \cdot & \cdot & \cdot & \cdot\\
\end{bmatrix}
&
\setlength\arraycolsep{3pt}
\begin{bmatrix}
\cdot & \cdot & \cdot & \cdot & \cdot & \cdot & \cdot & \cdot & \cdot\\
\cdot & \cdot & \cdot & \cdot & \cdot & \cdot & \cdot & \cdot & \cdot\\
\cdot & \cdot & \cdot & \cdot & \cdot & \cdot & \cdot & \cdot & \cdot\\
\cdot & \cdot & \cdot & \cdot & \cdot & \cdot & \cdot & \cdot & \cdot\\
\cdot & \cdot & \cdot & \cdot & \cdot & 1 & \cdot & 1 & 1\\
\cdot & \cdot & \cdot & \cdot & \cdot & 1 & \cdot & 1 & 1\\
\cdot & \cdot & \cdot & \cdot & \cdot & \cdot & \cdot & \cdot & \cdot\\
\cdot & \cdot & \cdot & \cdot & \cdot & 1 & \cdot & 1 & 1\\
\cdot & \cdot & \cdot & \cdot & \cdot & 1 & \cdot & 1 & 1\\
\end{bmatrix}
&
\setlength\arraycolsep{3pt}
\begin{bmatrix}
\cdot & \cdot & \cdot & \cdot & \cdot & \cdot & \cdot & \cdot & \cdot\\
\cdot & \cdot & \cdot & \cdot & \cdot & \cdot & \cdot & \cdot & \cdot\\
\cdot & \cdot & \cdot & \cdot & \cdot & \cdot & \cdot & \cdot & \cdot\\
\cdot & \cdot & \cdot & 1 & \cdot & \cdot & 1 & 1 & \cdot\\
\cdot & \cdot & \cdot & \cdot & \cdot & \cdot & \cdot & \cdot & \cdot\\
\cdot & \cdot & \cdot & \cdot & \cdot & \cdot & \cdot & \cdot & \cdot\\
\cdot & \cdot & \cdot & 1 & \cdot & \cdot & 1 & 1 & \cdot\\
\cdot & \cdot & \cdot & 1 & \cdot & \cdot & 1 & 1 & \cdot\\
\cdot & \cdot & \cdot & \cdot & \cdot & \cdot & \cdot & \cdot & \cdot\\
\end{bmatrix}
&
\setlength\arraycolsep{3pt}
\begin{bmatrix}
\cdot & \cdot & \cdot & \cdot & \cdot & \cdot & \cdot & \cdot & \cdot\\
\cdot & \cdot & \cdot & \cdot & \cdot & \cdot & \cdot & \cdot & \cdot\\
\cdot & \cdot & \cdot & \cdot & \cdot & \cdot & \cdot & \cdot & \cdot\\
\cdot & \cdot & \cdot & \cdot & \cdot & \cdot & \cdot & \cdot & \cdot\\
\cdot & \cdot & \cdot & \cdot & \cdot & \cdot & \cdot & \cdot & \cdot\\
\cdot & \cdot & \cdot & \cdot & \cdot & \cdot & \cdot & \cdot & \cdot\\
\cdot & \cdot & \cdot & \cdot & \cdot & \cdot & \cdot & \cdot & \cdot\\
\cdot & \cdot & \cdot & \cdot & 1 & 1 & \cdot & \cdot & 1\\
\cdot & \cdot & \cdot & \cdot & 1 & 1 & \cdot & \cdot & 1\\
\end{bmatrix}
\\
\setlength\arraycolsep{3pt}
\begin{bmatrix}
\cdot & \cdot & \cdot & \cdot & \cdot & \cdot & \cdot & \cdot & \cdot\\
\cdot & \cdot & \cdot & \cdot & \cdot & \cdot & \cdot & \cdot & \cdot\\
\cdot & \cdot & \cdot & \cdot & \cdot & \cdot & \cdot & \cdot & \cdot\\
1 & \cdot & \cdot & \cdot & \cdot & \cdot & 1 & \cdot & \cdot\\
1 & \cdot & \cdot & \cdot & \cdot & \cdot & 1 & \cdot & \cdot\\
\cdot & \cdot & \cdot & \cdot & \cdot & \cdot & \cdot & \cdot & \cdot\\
\cdot & \cdot & \cdot & \cdot & \cdot & \cdot & \cdot & \cdot & \cdot\\
\cdot & \cdot & \cdot & \cdot & \cdot & \cdot & \cdot & \cdot & \cdot\\
\cdot & \cdot & \cdot & \cdot & \cdot & \cdot & \cdot & \cdot & \cdot\\
\end{bmatrix}
&
\setlength\arraycolsep{3pt}
\begin{bmatrix}
\cdot & \cdot & \cdot & \cdot & \cdot & \cdot & \cdot & \cdot & \cdot\\
\cdot & \cdot & \cdot & \cdot & \cdot & \cdot & \cdot & \cdot & \cdot\\
\cdot & \cdot & \cdot & \cdot & \cdot & \cdot & \cdot & \cdot & \cdot\\
\cdot & \cdot & \cdot & \cdot & \cdot & \cdot & \cdot & \cdot & \cdot\\
\cdot & \cdot & 1 & \cdot & \cdot & \cdot & \cdot & 1 & 1\\
\cdot & \cdot & 1 & \cdot & \cdot & \cdot & \cdot & 1 & 1\\
\cdot & \cdot & \cdot & \cdot & \cdot & \cdot & \cdot & \cdot & \cdot\\
\cdot & \cdot & \cdot & \cdot & \cdot & \cdot & \cdot & \cdot & \cdot\\
\cdot & \cdot & \cdot & \cdot & \cdot & \cdot & \cdot & \cdot & \cdot\\
\end{bmatrix}
&
\setlength\arraycolsep{3pt}
\begin{bmatrix}
\cdot & 1 & \cdot & \cdot & 1 & \cdot & \cdot & \cdot & \cdot\\
\cdot & \cdot & \cdot & \cdot & \cdot & \cdot & \cdot & \cdot & \cdot\\
\cdot & 1 & \cdot & \cdot & 1 & \cdot & \cdot & \cdot & \cdot\\
\cdot & 1 & \cdot & \cdot & 1 & \cdot & \cdot & \cdot & \cdot\\
\cdot & \cdot & \cdot & \cdot & \cdot & \cdot & \cdot & \cdot & \cdot\\
\cdot & 1 & \cdot & \cdot & 1 & \cdot & \cdot & \cdot & \cdot\\
\cdot & \cdot & \cdot & \cdot & \cdot & \cdot & \cdot & \cdot & \cdot\\
\cdot & \cdot & \cdot & \cdot & \cdot & \cdot & \cdot & \cdot & \cdot\\
\cdot & \cdot & \cdot & \cdot & \cdot & \cdot & \cdot & \cdot & \cdot\\
\end{bmatrix}
&
\setlength\arraycolsep{3pt}
\begin{bmatrix}
\cdot & \cdot & \cdot & \cdot & \cdot & \cdot & \cdot & \cdot & \cdot\\
\cdot & \cdot & \cdot & \cdot & \cdot & \cdot & \cdot & \cdot & \cdot\\
\cdot & \cdot & \cdot & \cdot & \cdot & \cdot & \cdot & \cdot & \cdot\\
\cdot & \cdot & \cdot & \cdot & 1 & \cdot & \cdot & 1 & \cdot\\
\cdot & \cdot & \cdot & \cdot & \cdot & \cdot & \cdot & \cdot & \cdot\\
\cdot & \cdot & \cdot & \cdot & \cdot & \cdot & \cdot & \cdot & \cdot\\
\cdot & \cdot & \cdot & \cdot & 1 & \cdot & \cdot & 1 & \cdot\\
\cdot & \cdot & \cdot & \cdot & \cdot & \cdot & \cdot & \cdot & \cdot\\
\cdot & \cdot & \cdot & \cdot & 1 & \cdot & \cdot & 1 & \cdot\\
\end{bmatrix} 
\end{array}
\]
Hence we take each of these 12 matrices with a weight of $0.5$ and obtain that $\bpf(D\otimes D) \le 6 < 6.25 = \bpf^2(D)$, the square of the fractional biclique partition number of the Domino graph. Moreover, this decomposition is also optimal because there is a dual feasible solution certifying that $\bpf(D\ot D) \ge 6$. It is interesting to note that the degeneracy of the optimum solution increases from $D$ to $D \ot D$: while we have $7$ edges in $D$ and an optimum partition with $5$ bicliques, we have $49$ edges in $D \ot D$ and an optimum partition with $12$ bicliques. We will get back to this observation in Section~\ref{sec:upperbound}.

For now, it is more important that this example rules out that asymptotic fractional binary rank and fractional binary rank always coincide like asymptotic fractional boolean rank and fractional boolean rank always do. It remains the question what value has the asymptotic fractional binary rank of the Domino and does this give us a hint for its general behavior? We can only partially answer these questions as follows.

\begin{theorem}\label{thm:main}
    The asymptotic fractional binary rank, likewise the asymptotic fractional biclique partition number, of the Domino is at least $2$ and at most $2.372713$, i.e.,
    \[
    \bpf^\infty(D) \in [2,2.372713].
    \]
    In general, we have \[i(A) \le \bcf(A) \le \bpf^\infty(A) \le \bpf(A) \le \bp(A),\] i.e., the asymptotic fractional binary rank is sandwiched between the fractional boolean rank and the fractional bipartite rank, which are in turn sandwiched between the isolation-set number and the binary rank.
\end{theorem}

For all inequalities in the above chain except between $\bcf$ and $\bpf^\infty$, there are known examples where the inequality is strict. The Domino serves as an example for strictness of the last two inequalities as seen before and the Crown graph $\overline{I}_5$, i.e., a $5 \times 5$ all-ones matrix minus the identity, has $i(\overline{I}_5) = 3 < 10/3 = \bcf(\overline{I}_5)$ per~\cite{watts2006fractional}.

For us, it remains to prove the lower bound $\bcf(A) \le \bpf^\infty(A)$ in Section~\ref{sec:lowerbound} by showing a stronger lower bound for the fractional binary rank for Kronecker products. Furthermore, we obtain the numerical upper bound for the Domino by algorithm engineering as discussed in Section~\ref{sec:upperbound} based on the observation that, according to Fekete's Lemma, sub-multiplicativity not only implies the existence of the asymptotic fractional binary rank but also yields the following characterization:
\[
\bpf^\infty(A) := \lim_{k \to \infty} \sqrt[k]{\bpf\left(A^{\otimes k}\right)} = \inf \left\{\sqrt[k]{\bpf\left(A^{\otimes k}\right)}:\> k \in \mathbb{Z}_{>0} \right\}.
\]
Hence, any feasible fractional biclique partition of $D^{\ot k}$ yields an upper bound on $\bpf^\infty(D)$. We present our approach that allowed us to compute $\bpf(D^{\ot k}$ for $k$ up to $5$ in Section~\ref{sec:upperbound} and the obtained numerical upper bounds in Subsection~\ref{subsec:experiments}.

\section{Proof of the Lower Bound}\label{sec:lowerbound}

To prove the lower bound, we could simply exploit that $\bpf$ is lower bounded by $\bcf$ and the multiplicity of $\bcf$ with respect to the Kronecker product (cf.\ Theorem 5.3 in \cite{watts2006fractional}). However, such a proof does not provide much insight into the convergence of $\sqrt[k]{\bpf\left(A^{\otimes k}\right)}$ when $k$ goes to infinity. We therefore prove the following stronger Lemma for lower bounding the fractional biclique partition number of a Kronecker power.

\begin{lemma}\label{lem:lower}
For arbitrary binary matrices $A$ and $A'$, we have
\[\bpf(A\ot A')\geq \max\left\{\bcf(A)\cdot \bpf(A'),\> \bpf(A)\cdot \bcf(A') \right\}.\]
\end{lemma}

\begin{proof}
To prove this, we adopt the strategy from \cite{Lavynska2018}. We consider $\mathfrak{B}_{A\ot A'}$, $\mathfrak{B}_A$ and $\mathfrak{B}_{A'}$ as the sets of all bicliques in $G_{A\ot A'}$, $G_A$ and $G_{A'}$, respectively. 

We first prove that $\bpf(A\ot A')\geq \bcf(A)\cdot \bpf(A')$.
Since $\bpf$ is always positive, we can equivalently schow that 
\[
\bcf(A)\leq \frac{\bpf(A\ot A')}{\bpf(A')}.
\]
To prove this, let $w_{A\ot A'}$ be a weight assignment for bicliques from $\mathfrak{B}_{A\ot A'}$ that corresponds to the optimal fractional biclique partition of the graph $G_{A\ot A'}$ and the fractional biclique partition number $\bpf(A\ot A')$. For each biclique $B_{A\ot A'}\in \mathfrak{B}_{A\ot A'}$, the projection of $B_{A\ot A'}$ on $G_A$ is denoted by $\pi(B_{A\ot A'})$, which contains exactly those edges corresponding to $a_{ij} \neq 0$ of $A$, where at least one edge from the block $a_{ij}A'$ is contained in $B_{A\ot A'}$. From this definition of the projection, it follows that $\pi(B_{A\ot A'})$ is a biclique in $G_A$ for each $B_{A\ot A'}\in \mathfrak{B}_{A\ot A'}$.

Moreover, each biclique $B_A \in \mathfrak{B}_{A}$ collects the weight from all the bicliques that project to it, i.e., we set
\[
w_A(B_A):=\frac{1}{\bpf(A')}\left[\sum_{B_{A\ot A'}\in \pi^{-1}(B_A)}w_{A\ot A'}(B_{A\ot A'})\right]
\]
for all bicliques $B_A$ of $G_A$, where $\pi^{-1}(B_A) := \left\{B_{A\ot A'}\in \mathfrak{B}_{A\ot A'}: \pi(B_{A\ot A'})=B_A\right\}$ is the preimage of $B_A$ under $\pi$.

Thus, the weight of $B_A$ is an accumulated weight of all bicliques, projected on $B_A$ divided by $\bpf(A')$. The normalization with $\bpf(A')$ is chosen such that these weights form a fractional biclique cover for $G_A$. To see this, consider an arbitrary edge $e$ of $G_A$ corresponding to some $a_{ij} \neq 0$ of $A$. The intersection of the bicliques of $G_{A\ot A'}$ with the block $a_{ij}A'$ yields a feasible fractional biclique partition of $G_{A'}$ because the edges of $a_{ij}A'$ are fractional decomposed by those bicliques in $\mathfrak{B}_{A\ot A'}$ with a non-empty intersection with the block $a_{ij}A'$ and their corresponding weights as they come from a fractional biclique partition of the whole of $G_{A\ot A'}$. Hence, the sum of the weights of all bicliques intersecting the block $a_{i,j}A'$ is at least the optimum $\bpf(A')$.
Thus,
\begin{align*}
\bpf(A') &\le \sum_{B_{A\ot A'} \cap a_{ij}A'  \ne \emptyset} w_{A\ot A'}(B_{A\ot A'}) = \sum_{B_A: e \in B_A}\left[\sum_{B_{A\ot A'}\in \pi^{-1}(B_A)}w_{A\ot A'}(B_{A\ot A'})\right]
\end{align*}
Since each edge $e$ from $G_A$ corresponds to some $a_{ij}\neq 0$ of $A$, the sum of weights $w_A(B_A)$ from our construction over all bicliques containing $e$ of $G_A$ is at least $1$:
\begin{align*}
    \sum_{B_A:e\in B_A}w_A(B_A) &=\sum_{B_A:e\in B_A}\frac{1}{\bpf(A')}\left[\sum_{B_{A\ot A'}\in \pi^{-1}(B_A)}w_{A\ot A'}(B_{A\ot A'})\right] \\
    &\geq \frac{1}{\bpf(A')}\cdot \bpf(A')=1
\end{align*}
Hence, the assignment $w_A$ is a feasible fractional biclique cover of $G_A$.
This means that its objective value is at least the optimum, i.e.,
\begin{align*}
\bcf(A)&\leq \sum_{B_A\in \mathfrak{B}_A}w_A(B_A)=\sum_{B_A\in G_A}\frac{1}{\bpf(A')}\left[\sum_{B_{A\ot A'}\in \pi^{-1}(B_A)}w_{A\ot A'}(B_{A\ot A'})\right] \\
&=\frac{1}{\bpf(A')}\sum_{B_{A\ot A'}\in \mathfrak{B}_{A\ot A'}}w_{A\ot A'}(B_{A\ot A'}) =\frac{1}{\bpf(A')}\cdot \bpf(A\ot A'),
\end{align*}
which proves the claim. Observe that the role of $A$ and $A'$ are arbitrary, i.e., it also holds that $\bcf(A\ot A') \le \bpf(A)\cdot\bcf(A')$ by an analogous argument although the Kronecker product does not commute in general.
\end{proof}

Consider Lemma~\ref{lem:lower} with $A'=A^{\ot (k-1)}$. Then, we have
\[
\bpf(A^{\ot k}) = \bpf\left(A \ot A^{\ot (k-1)}\right)\geq \bpf(A) \cdot \bcf\left(A^{(k-1)}\right) = \bpf(A) \cdot \bcf^{k-1}(A)
\]
Hence,
\[
\bpf(A^{\otimes k})\geq \frac{\bpf(A)}{\bcf(A)} \cdot \bcf^k(A),
\]
and taking the $k^{\text{th}}$ root yields
\[
\sqrt[k]{\bpf(A^{\otimes k})} \geq \bcf(A)\cdot\sqrt[k]{\frac{\bpf(A)}{\bcf(A)}}.
\]

For the Domino with bipartite adjacency matrix $D$, we have the biclique cover number and the fractional biclique cover number as $2$ and the fractional biclique partition number as $2.5$. We have that the $k^{\text{th}}$ root of the fractional biclique number of the $k^{\text{th}}$ Kronecker power of $D$ as:
\[
\sqrt[k]{\bpf(A^{\otimes k})} \geq 2\cdot\sqrt[k]{\frac{2.5}{2}} = 2\cdot\left(\frac{5}{4}\right)^{1/k}.
\]
As $k$ becomes bigger, the lower bound for $\sqrt[k]{\bpf(A^{\otimes k})}$ approaches $2$, the fractional biclique cover number of $D$.

Note that Fekete's lemma does not say that the sequence $\sqrt[k]{\bpf(A^{\otimes k})}$ is monotonically decreasing - what it guarantees is that the sequence converges, and the limit of the sequence is the infimum of the sequence. The first values of the lower bound are approximately
\begin{center}
    \begin{tabular}{c|ccccc}
$k$ & 1 & 2 & 3 & 4 & 5\\ \hline
$2\cdot\left(\frac{5}{4}\right)^{1/k}$ & 2.5 & 2.236 & 2.154 & 2.115 & 2.091
\end{tabular}
\end{center}

Since the lower bound converges rather quickly to $2$, this gives some hope that we could get a good hint whether also $\sqrt[k]{\bpf(D^{\otimes k})}$ converges to $2$ by computing the fractional biclique partition numbers for the first few values for $k$, where every further data point may help us to narrow the gap by getting tighter bounds as we discuss in the next section.

\section{Algorithm for the Upper Bound}\label{sec:upperbound}

To compute the upper bound, we iteratively solve the linear optimization problem~\eqref{eqn:bpf} for $D^{\otimes k}$ and increasing $k$. We managed to solve it until $k=5$ on a machine with 256 GB RAM by the engineering described in this section. To this end, we used a Column Generation approach (except for the base case $k=1$ that is small enough to enumerate all bicliques explicitly). Observe that the number of edges grows as $7^k$ and the number of maximal bicliques grows as $4^k$. Since one of these maximal bicliques is of size $2^k \times 2^k$, we obtain $(2^{2^k}-1)\cdot(2^{2^k}-1)$ subbicliques already from this biclique, thus, the total number of bicliques is in $\Omega(4^{2^k})$. However, basic linear optimization theory tells us that there is always an optimum solution where the number of contributing bicliques is at most the number of edges. Moreover, we have already seen for $D$ and $D \otimes D$ that there are degenerate optimum solutions using even fewer bicliques. This means that in the end, we might be lucky to enumerate only a small fraction of all existing bicliques. This is where Column Generation comes into play. We refer the interested reader to~\cite{Luebbecke2011} to dive deeper into this topic, while we restrict ourselves to our adaptions of the basic idea of the general approach, which is as follows.
\begin{enumerate}
    \item Start with a small initial set of bicliques in the columns of $M$.
    \item \label{master}Solve the corresponding restricted LP (the so-called \emph{master problem}).
    \item Solve the so-called pricing problem to obtain new columns with negative reduced cost.
    \item\label{update} Stop if no such columns exist; otherwise update $M$ and start over with Step~\ref{master}.
\end{enumerate}

We discuss the details of these steps in the following subsections. Since we use a commercial solver for Step~\ref{master}, we include the corresponding details in Subsection~\ref{subsec:experiments}. Updating $M$ in Step~\ref{update} is discussed in Subsection~\ref{subsec:pruning} as it was also necessary to prune columns again to manage the size of the master problem.

\subsection{Initial Bicliques}\label{subsec:initial}

In principle, one could start with an empty set of bicliques and generate all the bicliques on the fly. However, this is not advisable as there would be little to no guidance for the pricing problem. Moreover, it simplifies the pricing problem if we assume that the initial set of bicliques allows for a feasible solution of the master problem.\footnote{This is not a necessary condition, but it greatly simplifies the exposition of the pricing problem in Subsection~\ref{subsec:pricing}.}

A trivial set of bicliques that guarantees feasibility are the individual edges. However, the resulting (fractional) biclique partition number would we far larger than the optimum. On the other hand, we consider a fixed side of the bipartite graph and combine all edges that are incident to the same node in a star biclique. This also yields a biclique partition. Note that these bicliques correspond to either the rows or the columns of the given matrix representation.

For small graphs, like the Domino, we can even afford to enumerate all bicliques. 

For Kronecker powers, we can generate good initial sets of bicliques inductively. Given a feasible set of bicliques $\mathcal{B}_1$ for $\bpf(A)$ and a feasible set of bicliques $\mathcal{B}_k$ for $\bpf(A^{\otimes k})$, we obtain a feasible set of bicliques $\mathcal{B}_{k+1} := \left\{ B \otimes B': B \in \mathcal{B}_1 \, \wedge \, B' \in \mathcal{B}_k\right\}$ for $\bpf(A^{\otimes (k+1)})$.

In particular, if we use the sets of bicliques that have a positive weight in the optimum solutions for $\bpf(A)$ and $\bpf(A^{\otimes k})$, respectively, we obtain not only a feasible initial set of bicliques for the initial master problem, it also implies that its optimum objective value is at most $\bpf(A)\cdot \bpf(A^{\otimes k})$, which is likely not optimal, but a good educated guess. The lower the initial objective value of the initial master problem, the fewer iterations will typically be necessary to find the optimum.

\subsection{Solving the Pricing Problem}\label{subsec:pricing}

The goal of the pricing problem is to identify one or more bicliques that such that adding them to the master problem will improve its current objective value. As described in Subsection~\ref{subsec:initial}, we may assume to start with a feasible initial master problem. One could also start with an infeasible one and use Farkas pricing instead~\cite{Luebbecke2011}. Hence, the master problem and its dual have finite optimum solutions and their objective values coincide. For the pricing problem, we use the optimum dual solution represented by a vector $y^*$ with an entry for each edge of the graph under consideration. Recall that some of these values might also be negative.

Before solving the first pricing problem, we compute the set of inclusion-wise maximal bicliques. To obtain this, we can use the MBEA algorithm from \cite{Zhang2014}. If we are given a Kronecker power $A^{\otimes k}$, it suffices to run that algorithm on $A$ and then take the $k$-th Kronecker powers of the inclusion-wise maximal bicluqes of $A$.

Let $\mathsf{B_1},\mathsf{B_2},\dots$ be the set of inclusion-wise maximal bicliques.
We iterate over all these bicliques, where we consider at Step $t$ the biclique $\mathsf{B_t}$.

We then generate new bicliques similar to the technique used in finding the submatrix with largest sum~\cite{derval2022maximal, yampc2021submatrix}. The submatrix with largest sum problem involves finding a submatrix of a data matrix $y^*_{i,j}$ such that the sum of the elements is maximized. A submatrix can include contiguous or non-contiguous rows and columns. This problem can be formalized as 
\[\max \left\{ \sum_{i,j} y^*_{i,j}\cdot r_i\cdot c_j: \> r_i,c_j\in\{0,1\} \right\}, \]
which is a quadratic unconstrained boolean optimization (QUBO) problem.
Such problems can be linearized through introducing variable $s_{i,j}\in\{0,1\}$. The new problem is:
$$\max\sum_{i,j} y^*_{i,j}\cdot s_{i,j}$$
$$s_{i,j}\leq r_i \forall i,j$$
$$s_{i,j}\leq c_j \forall i,j$$
$$s_{i,j}\geq r_i + c_j -1 \forall i,j$$
$$r_i,c_j\in\{0,1\}, s_{i,j}\in \{0,1\}$$
For every maximal biclique in the set of maximal bicliques, we formulate a new  integer linear program (ILP). All these ILPs can be solved independently in parallel, which we did on a machine with 48 cores and 256 GB RAM.

First, we only collected the optimum solutions for all inclusion-wise maximal bicliques in a pool provided that the respective objective value exceeds $1+\varepsilon$ for numerical reasons\footnote{We chose $\varepsilon = 10^{-6}$.}. Moreover, we discarded duplicates.

Later, when we established pruning (see Sec.~\ref{subsec:pruning}), we could afford to be more aggressive in the collection of new bicliques, so we included all bicliques from Gurobi's solution pool whose objective value exceeded the threshold into our biclique pool.

We compare the result of this new linear program (maximizing the sum of the dual variable values) to the thresholds ($1+\epsilon)$ if the primal solution is feasible, and $\epsilon$ if the primal solution is infeasible), where $\epsilon$ is an arbitrarily small quantity, e.g., $10^{-6}$. If the result is greater than the threshold, we can improve our solution by adding a new biclique. The Boolean variables corresponding to the edges that are set to True in this linear programming solution are guaranteed to form a biclique as well.

All bicliques form our biclique pool were added to the master problem to start over with the next iteration of the column generation process. If the pool is empty, we are done and we have found the last objective value of the master problem is the fractional biclique partition number. To be more precise, it is better than a $(1+\varepsilon)$-approximation since the dual solution scaled down by $1+\varepsilon$ is feasible and hence the duality gap is at most $1+\varepsilon$.

In fact, we record the maximum objective value, say $\alpha$, over all Integer Linear Programs for each inclusion-wise maximal bicliques. By scaling the optimum objective value of the master problem down by $\alpha$, we obtain a lower bound for the fractional biclique partition number. We kept track of the maximum lower bound over the iterations of the column generation process to monitor the progress.

\subsection{Pruning}\label{subsec:pruning}

To ensure that the number of bicliques does not blow up too fast when we 
It is possible that the procedure for generating new bicliques generates a lot of excess bicliques that do not have significant weights. Hence, we also track, for each biclique, the number of times its corresponding entry after calculating the dual solution $M^T y\leq \mathds{1}$ is slack, i.e., is less than $1$. If this occurs for more than $3$ iterations (a heuristic we choose), we prune this biclique from the set of bicliques used.

This is analogous to the pivoting decisions made in the simplex algorithm, where we consider the reduced cost of the variable to determine whether we pivot that variable to the basis. If the reduced cost is positive for a non-basic variable, adding the variable does not help the dual problem (a maximization problem). Instead of deleting all such bicliques, we allow the counter to go to $3$ before we delete the column corresponding to that biclique from our matrix $M$. This is to manage the size of the master problem so that it still fits into main memory and computing its optimum solution is not slowed down by surplus variables that would not be used in an optimum solution anyway.  Note that there is no guarantee that a biclique one pruned will not be regenerated later - the heuristic merely helps speed up the algorithm.

\subsection{Experimental Results}\label{subsec:experiments}

We ran our experiments on a compute server with 24 physical cores and 48 logical Intel(R) Xeon(R) CPU E5-2680 v3 processors at 2.50GHz (instruction set [SSE2|AVX|AVX2]) and 256 GB RAM. Our code is Open Source (GNU General Public License) and publicly available.\footnote{\url{https://github.com/AngikarGhosal/FractionalBicliquePartitionBinaryRank}}
However, it relies on Gurobi\footnote{\url{https://www.gurobi.com/}} as solver for the master problem and the sub-problems for the pricing problem. Although, it is a commercial solver, there is are academic licenses. In principle, it could be replaced by any other solver for (integer) linear optimization problems, potentially effecting the performance. We ran our experiments with version 10.0.0 of Gurobi and it used up to 24 threads. It used its concurrent optimizer, but the barrier optimizer (including crossover) was always the fastest.

We calculate the fractional biclique partition numbers (rounded to six decimal digits) of the first $5$ Kronecker powers of the Domino:
\begin{center}
    \begin{tabular}{c|ccccc}
$k$ & 1 & 2 & 3 & 4 & 5\\ \hline
$\bpf(D^{\otimes k})$ & 2.5 & 6 & 13.818792 & 32.040389 & 75.201302 \\
$\sqrt[k]{\bpf(D^{\otimes k})}$ & 2.5 & 2.449490 & 2.399699 & 2.379164 & 2.372712 \\
\end{tabular}
\end{center}

To obtain the optimum value for $k = 5$, we had to complete $1456$ iterations of Column Generation. Before we implemented the pruning mechanism, one iteration for $k=5$ took about $1.5$ hours, afterwards it dropped to about $50$min on average. This allowed us to be more aggressive when collecting new bicliques: instead of taking only the maximum-weight subbiclique for each maximal biclique, we then took all subbicliques exceeding the threshold that entered Gurobis solution pool. This significantly improved the progress of the objective value of the master problem. In the last $20$ iterations, solving the master problem became again very difficult due the crossover, which took up to $8$h for the final iteration after the barrier was already done in about $5$min.
The theory of linear optimization tells us that the fractional biclique partition number must be a rational number. A continued fraction approximation for the computed value of $\bpf{D^{\otimes 3}}$ suggests that it could be $2059/149$, however, we do not have a combinatorial interpretation for this ratio.

On the other hand, the upper bounds do not decrease as quickly as the lower bounds from~\ref{sec:lowerbound}, in particular, the last step from $k=4$ to $k=5$ only changed the third decimal digit of $\sqrt[k]{\bpf(D^{\otimes k})}$. This does not rule out a convergence to $2$, but it comes as a surprise, which asks for the solution for $k=6$. However, our computational resources were not enough to even generate the initial master problem. In the future, other researchers with more computational resources could take it from here and compute more datapoints.

\bibliography{my_ref}
\end{document}